\documentclass{article}

\usepackage{arxiv}

\usepackage[utf8]{inputenc}
\usepackage[T1]{fontenc}
\usepackage{microtype}

\usepackage{amsmath,amssymb,amsfonts,amsthm,mathtools,bm}
\numberwithin{equation}{section}

\usepackage{booktabs}
\usepackage{enumitem}
\usepackage[numbers,sort&compress]{natbib}
\usepackage[hidelinks]{hyperref}
\usepackage[nameinlink,noabbrev]{cleveref}
\usepackage{url}
\usepackage{doi}

\hypersetup{
pdftitle={
Contextual Fraction on Permutation Gain Graphs:
Exact Algorithms, Query Lower Bounds, and Dynamic Maintenance
},
pdfauthor={Ronald Katende},
pdfsubject={
Exact and dynamic computation of contextual fraction on
permutation gain graphs
},
pdfkeywords={
contextual fraction,
permutation gain graphs,
holonomy,
exact algorithms,
query complexity,
dynamic data structures,
constraint satisfaction
}
}

\newcommand{\D}{\mathcal{D}}
\newcommand{\supp}{\operatorname{supp}}
\newcommand{\Fix}{\operatorname{Fix}}
\newcommand{\Sym}{\operatorname{Sym}}
\newcommand{\id}{\operatorname{id}}
\newcommand{\CF}{\operatorname{CF}}
\newcommand{\NCF}{\operatorname{NCF}}
\newcommand{\CSP}{\operatorname{CSP}}
\newcommand{\Mov}{\operatorname{Mov}}
\newcommand{\pushf}[2]{\left(#1\right)_{\!*}#2}
\newcommand{\1}{\mathbf{1}}

\newtheorem{theorem}{Theorem}[section]
\newtheorem{lemma}[theorem]{Lemma}
\newtheorem{proposition}[theorem]{Proposition}
\newtheorem{corollary}[theorem]{Corollary}

\theoremstyle{definition}
\newtheorem{definition}[theorem]{Definition}
\newtheorem{example}[theorem]{Example}

\theoremstyle{remark}
\newtheorem{remark}[theorem]{Remark}
\newtheorem{openproblem}[theorem]{Open Problem}

\title{
Contextual Fraction on Permutation Gain Graphs:
Exact Algorithms, Query Lower Bounds, and Dynamic Maintenance
}

\author{
Ronald Katende \\
Department of Mathematics \\
Kabale University \\
Kikungiri Hill, Katuna Road, Kabale, Uganda \\
\texttt{rkatende@kab.ac.ug}
}

\date{}

\renewcommand{\headeright}{}
\renewcommand{\undertitle}{}
\renewcommand{\shorttitle}{
Contextual Fraction on Permutation Gain Graphs
}

\begin{document}

\maketitle

\begin{abstract}
For an explicitly represented finite empirical model, deciding whether the contextual fraction is strictly below one is NP-complete, while the standard exact linear program has one column for every global assignment. We identify a permutation-transport class in which this global problem collapses to a fixed-point calculation. Let a connected permutation gain graph act on a finite state set $O$, let $H\leq\Sym(O)$ be its holonomy subgroup, let $F=\Fix(H)$, and let $p$ be an $H$-invariant root distribution. For the induced empirical model,
\[
\NCF(e)=p(F),\qquad \CF(e)=1-p(F).
\]
Consequently, compatibility, $F$, and $\CF(e)$ are computable in $O(|O|(|V|+|E|))$ arithmetic and table operations. For every finite simple $2$-edge-connected graph, any deterministic exact algorithm in the explicit permutation-table query model requires at least $(|O|-1)|E|$ probes in the worst case, making the dependence on the input tables optimal up to constant factors. With a fixed spanning tree, chord insertions and deletions require $O(|O|)$ worst-case time, or time proportional to the moved-set representation, while compatibility and contextual-fraction queries take $O(1)$ time. Finally, for common-marginal realizable binary constraint languages, the support threshold $\CF<1$ is polynomial-time equivalent to the associated finite-domain constraint-satisfaction problem and therefore inherits the Bulatov--Zhuk dichotomy. The results identify a query-optimal and dynamically maintainable tractability island inside the general contextual-fraction problem.
\end{abstract}


\section{Introduction}
\label{sec:introduction}

A compatible family of local probability distributions need not be the family
of marginals of one joint distribution.  This is the marginal-extension
problem \cite{Vorobev1962}; in the sheaf-theoretic framework, failure of a
global extension is contextuality \cite{AbramskyBrandenburger2011}.  The
contextual fraction is the complement of the largest noncontextual
subprobability that can be embedded in the model
\cite{AbramskyBarbosaMansfield2017}.  Its standard finite linear program has
one variable for every global assignment, so direct computation is exponential
in the number of vertices.

The exponential number of LP columns is not merely a representational
inconvenience. At the support level,
\[
\CF(e)<1
\quad\Longleftrightarrow\quad
\NCF(e)>0
\quad\Longleftrightarrow\quad
\exists g\in O_X\ \text{such that}\
e_C(g|_C)>0\ \text{for every }C\in\mathcal M.
\]
Thus deciding whether an empirical model is not maximally contextual is a
global constraint-satisfaction problem. As shown formally in
\cref{prop:general-complexity}, this decision problem is NP-complete even for
binary contexts with three outcomes, while deciding $\CF(e)=1$ is
coNP-complete. The contribution below is therefore not only a compact
reformulation of a large LP: it identifies a class on which a generally
intractable support problem, and the exact quantitative value beyond it,
become efficiently computable.

The deterministic structure used here is classical permutation gain-graph
theory.  Oriented edges carry permutations, spanning-tree normalization
produces fundamental holonomies, and satisfying states correspond to common
fixed points of the holonomy action
\cite{GrossTucker1987,Zaslavsky1989,Zaslavsky2009,ZaslavskyGlossary}.
Holonomy has also appeared in bundle formulations of contextuality
\cite{Montanhano2021}.  Those constructions are inherited.

\subsection{Computational context and relation to prior work}

The connection between contextuality, global sections, relational hidden
variables, and constraint satisfaction predates the present work
\cite{Abramsky2013,AbramskyGottlobKolaitis2013,Simmons2018}.
Abramsky, Gottlob, and Kolaitis studied the complexity of robust
constraint-satisfaction formulations arising from local hidden-variable
models, while Simmons established NP-completeness results for detecting
possibilistic locality in restricted measurement scenarios. These works
concern support extendability and related hidden-variable decision problems.
They do not provide the probability-weighted fixed-set identity proved here,
a matching query lower bound for exact contextual-fraction evaluation, or a
dynamic data structure for maintaining the exact value under gain updates.

The contextual fraction itself was introduced together with its primal and
dual linear programs and its resource-theoretic interpretation
\cite{AbramskyBarbosaMansfield2017}. The present paper does not modify that
definition. It identifies a structured empirical-model class for which the
generic LP optimum has a closed form and then studies the exact static,
query, dynamic, and support complexity of computing that form.

The paper studies the probability-weighted empirical model obtained by
sampling one root state and transporting it along the gain graph.  For an
$H$-invariant root distribution $p$, the central reduction is
\begin{equation}
\label{eq:intro-main}
\NCF(e)=p(\Fix(H)),\qquad
\CF(e)=1-p(\Fix(H)).
\end{equation}
This converts the contextual-fraction LP into a fixed-set-mass computation.
For example, with eight states on a one-hundred-vertex cycle, the general LP
has
\[
8^{100}\approx2.04\times10^{90}
\]
global-assignment columns, whereas the structural representation has scale
\[
|O|(|V|+|E|)=1600.
\]
The comparison concerns exact representation and computation, not statistical
sampling in a high-dimensional Euclidean space.

\subsection{Contributions and originality}

The paper delivers four results.

\begin{enumerate}
\item It proves \eqref{eq:intro-main} for the probability-weighted empirical
model defined in \cref{sec:model}.  Classical gain-graph theory identifies the
fixed set; the paper-specific step is proving that its probability mass is
exactly the LP-optimal noncontextual weight.

\item It gives an $O(k(n+m))$ exact algorithm and a matching deterministic
query lower bound on $2$-edge-connected graphs.  Thus the algorithm is not
merely efficient relative to the exponential LP; it is worst-case optimal in
the explicit-table model.

\item It gives a semi-dynamic data structure under fixed-tree chord updates.
Dense insertions and deletions take $O(k)$ time, sparse holonomies take
$O(|\Mov(h)|)$ time, and compatibility and contextual-fraction queries take
$O(1)$ time.  Fully dynamic tree-edge updates are isolated as an open problem
rather than assigned an unsupported bound.

\item It gives a compatibility-preserving probabilistic encoding of every
instance over a common-marginal realizable binary language and proves that the
resulting support-threshold problem is polynomial-time equivalent to
$\CSP(\Gamma\cup\{=\})$.  Consequently, the Bulatov--Zhuk dichotomy transfers
exactly to this promised empirical-model class.  The claim is the
compatibility-preserving encoding and the resulting restricted transfer, not
the pre-existing general connection between contextuality and constraint
satisfaction.
\end{enumerate}

These results form one computational statement: permutation transport is a
natural structural restriction that collapses an exponential global search to
an optimal linear-time algorithm, remains maintainable under local cycle
updates, and sits on a sharply characterized support-complexity boundary.

\section{Contextual Fraction as a Finite Linear Program}
\label{sec:background}

\begin{definition}[Measurement scenario and empirical model]
A \emph{measurement scenario} is a triple $(X,\mathcal M,O)$, where $X$ is a
finite set of observables, $\mathcal M$ is a cover of $X$ by contexts, and
$O=\{O_x\}_{x\in X}$ assigns a finite outcome set to each observable.  For
$U\subseteq X$, write $O_U=\prod_{x\in U}O_x$.  An \emph{empirical model} is
a family $e=\{e_C\in\D(O_C)\}_{C\in\mathcal M}$ satisfying
\[
e_C|_{C\cap C'}=e_{C'}|_{C\cap C'}
\qquad(C,C'\in\mathcal M).
\]
It is \emph{noncontextual} if there exists $d\in\D(O_X)$ with
$d|_C=e_C$ for every $C$.
\end{definition}

\begin{definition}[Contextual fraction]
The noncontextual and contextual fractions are
\begin{align*}
\NCF(e)
&:=\max\left\{\lambda\in[0,1]:
e=\lambda e^{\mathrm{NC}}+(1-\lambda)e',\quad
e^{\mathrm{NC}}\text{ noncontextual}
\right\},\\
\CF(e)&:=1-\NCF(e).
\end{align*}
\end{definition}

Index rows by local events $(C,s)$ and columns by global assignments
$g\in O_X$.  Define
\[
M_{(C,s),g}:=\1[g|_C=s],
\qquad
(v_e)_{(C,s)}:=e_C(s).
\]

\begin{proposition}[Primal and dual forms]
\label{prop:lp}
For every finite empirical model,
\begin{align}
\NCF(e)
&=\max\{\1^{\mathsf T}b:Mb\leq v_e,\ b\geq0\},
\label{eq:primal}\\
&=\min\{v_e^{\mathsf T}y:M^{\mathsf T}y\geq\1,\ y\geq0\}.
\label{eq:dual}
\end{align}
\end{proposition}

\begin{proof}
If $e=\lambda e^{\mathrm{NC}}+(1-\lambda)e'$ and $q$ is a global
distribution for $e^{\mathrm{NC}}$, then $b=\lambda q$ is feasible in
\eqref{eq:primal} with objective $\lambda$.  Conversely, let $b$ be feasible
and put $\lambda=\1^{\mathsf T}b$.  The entries of $Mb$ belonging to any
fixed context sum to $\lambda$, so $0\leq\lambda\leq1$.  For $0<\lambda<1$, the vector $q=b/\lambda$ is a global distribution.  The
vector $Mb$ is the family of local marginals of the global subdistribution
$b$, and is therefore compatible.  Hence
\[
v_{e'}:=\frac{v_e-Mb}{1-\lambda}
\]
is nonnegative, normalized in every context, and compatible.  Hence
$e=\lambda e^{\mathrm{NC}}+(1-\lambda)e'$.  The endpoint cases are
immediate.  This proves \eqref{eq:primal}; \eqref{eq:dual} follows from
finite-dimensional LP duality.
\end{proof}

\begin{proposition}[General support-threshold complexity]
\label{prop:general-complexity}
Consider the following decision problem: given an explicitly represented finite
compatible empirical model with rational entries, decide whether
$\CF(e)<1$. This problem is NP-complete, even when every context has size two
and every observable has the common outcome set $\{1,2,3\}$. Consequently,
deciding whether $\CF(e)=1$ is coNP-complete, and exact evaluation of
$\CF(e)$ is NP-hard.
\end{proposition}

\begin{proof}
By the primal formulation, $\NCF(e)>0$ if and only if some global assignment
$g$ selects a positive-probability event in every context. Such a global
assignment is a polynomial-size certificate, so the problem belongs to NP.

For NP-hardness, reduce graph $3$-colourability. Given a graph
$G=(V,E)$, take one observable for each vertex, outcome set
$O=\{1,2,3\}$, and one context $\{u,v\}$ for each edge. Define
\[
e_{uv}(a,b):=\frac{1}{6}\,\1[a\neq b].
\]
Every one-vertex marginal is uniform, so the resulting empirical model is
compatible. A global assignment selects a positive event on every edge if
and only if it is a proper $3$-colouring of $G$. Therefore
\[
G\text{ is $3$-colourable}
\quad\Longleftrightarrow\quad
\NCF(e)>0
\quad\Longleftrightarrow\quad
\CF(e)<1.
\]
The remaining claims follow by complementation and by observing that exact
evaluation decides this threshold.
\end{proof}

\section{Probability-Weighted Permutation Gain Graphs}
\label{sec:model}

Let $G=(V,E)$ be a finite connected simple undirected graph with
$n=|V|\geq2$ and $m=|E|$, let $O$ be a finite state set, and let
$k=|O|$.  Each oriented edge $u\to v$ carries a permutation
$\sigma_{uv}\in\Sym(O)$ with $\sigma_{vu}=\sigma_{uv}^{-1}$.  A state
$g:V\to O$ satisfies the edge when
\[
g(v)=\sigma_{uv}(g(u)).
\]
This is the standard permutation-action form of a gain-graph state
\cite{Zaslavsky2009}.

Choose a root $r$ and a spanning tree $T$.  Let $\tau_v$ be the product of
edge gains along the unique tree path from $r$ to $v$, with $\tau_r=\id$.
For each oriented edge define
\begin{equation}
\label{eq:holonomy}
h_{uv}:=\tau_v^{-1}\sigma_{uv}\tau_u.
\end{equation}
Tree edges have identity holonomy.  The non-tree holonomies generate a subgroup
$H\leq\Sym(O)$, and
\[
F:=\Fix(H)=\{a\in O:h(a)=a\text{ for every }h\in H\}.
\]
The spanning-tree construction and the fixed-point description of satisfying
states are classical; they are restated only to fix notation.

Let $p\in\D(O)$ and define
\[
p_v:=\pushf{\tau_v}{p}.
\]
On an oriented edge $u\to v$, define
\begin{equation}
\label{eq:edge-model}
e_{uv}(x,y)
:=p_u(x)\,\1[y=\sigma_{uv}(x)].
\end{equation}

\begin{proposition}[Compatibility]
\label{prop:compatibility}
The family \eqref{eq:edge-model} is a compatible empirical model on the edge
cover of $G$ if and only if $p$ is invariant under $H$.
\end{proposition}

\begin{proof}
The marginal of $e_{uv}$ at $u$ is $p_u$.  Its marginal at $v$ is
\[
\pushf{\sigma_{uv}}{p_u}
=\pushf{\sigma_{uv}\tau_u}{p}
=\pushf{\tau_v}{\pushf{h_{uv}}{p}},
\]
using \eqref{eq:holonomy}.  Since pushforward by $\tau_v$ is injective, this
equals $p_v=\pushf{\tau_v}{p}$ exactly when
$\pushf{h_{uv}}{p}=p$.  The edge cover has singleton overlaps, so compatibility
is equivalent to this condition for every generator, hence for $H$.
\end{proof}

\begin{lemma}[Global states]
\label{lem:global-states}
The edge-satisfying states are in bijection with $F$ through
\[
a\longmapsto g_a,
\qquad
g_a(v)=\tau_v(a).
\]
Moreover, $g_a$ is support-compatible with \eqref{eq:edge-model} if and only
if $a\in F\cap\supp(p)$.
\end{lemma}

\begin{proof}
A satisfying state is determined by its root value $a=g(r)$, and propagation
along $T$ forces $g(v)=\tau_v(a)$.  The relation on an arbitrary edge is then
\[
\tau_v(a)=\sigma_{uv}\tau_u(a),
\]
equivalent to $h_{uv}(a)=a$.  Thus all edge relations hold exactly for
$a\in F$.  For such $a$, the selected event on every edge has probability
$p(a)$, proving the support statement.
\end{proof}

\begin{theorem}[Probability-weighted contextual-fraction identity]
\label{thm:main}
Assume the compatibility condition in \cref{prop:compatibility}.  Then
\begin{equation}
\label{eq:main}
\NCF(e)=p(F),\qquad \CF(e)=1-p(F).
\end{equation}
\end{theorem}

\begin{proof}
Let $b$ be feasible in \eqref{eq:primal}.  A global assignment that violates
an edge relation selects a zero-probability row and therefore has weight zero.
By \cref{lem:global-states}, the only remaining columns are $g_a$ with
$a\in F\cap\supp(p)$.  Fix such an $a$.  On any edge, $g_a$ selects an event
of probability $p(a)$, and no $g_{a'}$ with $a'\neq a$ selects the same event
because $\tau_u$ is injective.  Hence $b_{g_a}\leq p(a)$, and
\[
\1^{\mathsf T}b\leq\sum_{a\in F}p(a)=p(F).
\]

For the reverse inequality, set $b_{g_a}=p(a)$ for $a\in F$ and all other components to zero.  Distinct root values select distinct events at each edge, so every row load is either its full empirical capacity or zero.  Thus $b$ is feasible and has objective $p(F)$. Moreover the exponentially many primal variables are not merely unnecessary computationally: the displayed optimum is supported on at most $|F|\leq k$ global assignments.

\end{proof}

\begin{corollary}[Increment under an added chord]
\label{cor:increment}
Let a compatible model have holonomy subgroup $H_0$ and fixed set
$F_0=\Fix(H_0)$.  Add a chord whose fundamental holonomy is $h$.  If $p$ is
invariant under both the old and new holonomy groups, then
\[
F=F_0\cap\Fix(h)
\]
and
\begin{equation}
\label{eq:increment}
\CF(e)-\CF(e_0)
=p\bigl(F_0\setminus\Fix(h)\bigr).
\end{equation}
\end{corollary}

\begin{proof}
The new group is $\langle H_0,h\rangle$, whose fixed set is
$F_0\cap\Fix(h)$.  Apply \cref{thm:main} to both models and subtract.
\end{proof}

\begin{remark}[Classical consequences]
Trees, switching invariance, independence of the spanning-tree choice, and the
usual odd-cycle obstruction follow immediately from the classical holonomy
representation.  They are not developed as separate sections.  In particular,
a tree has trivial holonomy and $\CF=0$, while a uniform root distribution
gives $\CF=1-|F|/|O|$.
\end{remark}

\section{Static Evaluation and Query Optimality}
\label{sec:static}

Fix one stored orientation for every undirected edge. Its gain is represented by an explicit array of length $k$, and a table probe specifies an edge and a domain state and returns the corresponding array entry. Reverse tables and all tree transports may be constructed explicitly in $O(k)$ operations per edge. We work in a unit-cost RAM model for table access, comparisons, and exact arithmetic on entries of $p$. Thus the bounds below count arithmetic and table operations; in a bit-complexity model, the encoding lengths of the entries of $p$ contribute an additional arithmetic cost.

\begin{proposition}[Exact static evaluation]
\label{prop:static-upper}
Compatibility, the common fixed set $F$, and, whenever compatibility holds,
the exact value $\CF(e)$ can be computed in
\[
O\bigl(k(n+m)\bigr)
\]
arithmetic and table operations using $O(kn)$ working space, excluding the
read-only input tables.
\end{proposition}

\begin{proof}
Construct a rooted spanning tree in $O(n+m)$ time.  Compute all tree transports
and inverses using one length-$k$ permutation composition per tree edge, for
$O(kn)$ time.  For each non-tree edge, compute its fundamental holonomy
pointwise, test $p(h(a))=p(a)$ for every $a$, and intersect its fixed set into
a running set $F$.  This costs $O(k)$ per non-tree edge and $O(km)$ in total.
If compatibility holds, sum $p(a)$ over $a\in F$ and apply
\cref{thm:main}.
\end{proof}

The next theorem gives a genuine matching lower bound.  The restriction to
$2$-edge-connected graphs is necessary: a bridge lies on no cycle, so its
gain cannot affect the holonomy fixed set.

\begin{theorem}[Worst-case query lower bound]
\label{thm:static-lower}
Let $k\geq2$, and let $G=(V,E)$ be a finite simple $2$-edge-connected graph with $n$ vertices and $m$ edges. Consider uniform root distribution on $O$ and arbitrary permutation gains presented as explicit tables.  Every deterministic algorithm that outputs the exact contextual fraction for every such input must make at least
\[
(k-1)m
\]
table probes in the worst case.  Consequently, because a $2$-edge-connected graph satisfies $m\geq n$, the worst-case complexity is
\[
\Omega\bigl(k(n+m)\bigr).
\]
\end{theorem}

\begin{proof}
Run the algorithm on the input in which every edge gain is the identity. This input is compatible, $H$ is trivial, and $\CF=0$.

Suppose the algorithm makes fewer than $(k-1)m$ probes.  Then some edge $e$ has at most $k-2$ queried domain points.  Choose two unqueried states $a,b\in O$.  Form a second input by leaving every other edge gain equal to the identity and replacing the gain on $e$ by the transposition $(a\ b)$. The two inputs agree on every table entry probed by the algorithm.

Because $G$ is $2$-edge-connected, $e$ lies on a cycle.  The gain around that cycle in the second input is $(a\ b)$ or its inverse, which is the same transposition.  All other gains are identities, so the holonomy subgroup is
$\{\id,(a\ b)\}$ and
\[
F=O\setminus\{a,b\}.
\]
Uniformity makes both inputs compatible. By \cref{thm:main}, the second input has
\[
\CF=1-\frac{k-2}{k}=\frac{2}{k},
\]
whereas the first has $\CF=0$.  Since the algorithm receives identical answers to all its probes, it cannot be correct on both inputs.  Therefore, on the all-identity input, every correct deterministic exact algorithm must probe at least $k-1$ domain points of every edge, for a total of at least $(k-1)m$ probes.
\end{proof}

\begin{corollary}[Worst-case optimality]
\label{cor:static-optimal}
On explicit-table inputs, the algorithm of \cref{prop:static-upper} is
worst-case optimal up to constant factors on $2$-edge-connected graphs.
\end{corollary}

\begin{remark}
The lower bound is for deterministic exact algorithms.  It does not claim an
instance-optimal bound on graphs with bridges, nor a lower bound for randomized
approximation algorithms.
\end{remark}

\section{Exact Maintenance under Fixed-Tree Chord Updates}
\label{sec:dynamic}

Fix a root, a spanning tree $T$, its transports $\tau_v$, and a root distribution $p$. Tree edges remain fixed.  Non-tree edges may be inserted or deleted. Each active chord has a unique identifier. An insertion supplies a chord not currently active together with its gain, and a deletion specifies the identifier of a previously inserted chord. Let $m_0$ denote the number of edges present during preprocessing. For an active chord $e=(u,v)$, let
\[
h_e=\tau_v^{-1}\sigma_{uv}\tau_u,
\qquad
\Mov(h_e)=\{a\in O:h_e(a)\neq a\}.
\]

Maintain an integer counter
\[
c(a):=|\{e\text{ active}:a\in\Mov(h_e)\}|,
\]
a running total
\[
W:=\sum_{a:c(a)=0}p(a),
\]
and an integer $B$ equal to the number of active chords whose holonomy does
not preserve $p$.  For each active chord, cache $\Mov(h_e)$ and whether it is
$H$-invariance violating.

\begin{theorem}[Chord-dynamic data structure]
\label{thm:dynamic}
After $O(k(n+m_0))$ preprocessing, the data structure supports:
\begin{enumerate}
\item insertion of a chord given by a dense permutation table in $O(k)$ worst-case time;
\item deletion of a cached chord in $O(k)$ worst-case time;
\item a  compatibility-and-contextual-fraction query in $O(1)$ time.
\end{enumerate}
If the fundamental holonomy is supplied sparsely as the list of pairs
\[
\{(a,h_e(a)):a\in\Mov(h_e)\},
\]
then insertion and deletion take $O(|\Mov(h_e)|)$ time, including the compatibility test. A query reports ``incompatible'' when $B>0$ and otherwise
returns
\[
\CF(e)=1-W.
\]
\end{theorem}

\begin{proof}
A state $a$ belongs to the current common fixed-point set exactly when every
active holonomy fixes it, equivalently when $c(a)=0$.  Thus $W=p(F)$.

On insertion, compute $h_e$ and its moved set.  For every
$a\in\Mov(h_e)$, increment $c(a)$; if the counter changes from $0$ to $1$,
subtract $p(a)$ from $W$.  Test $p(h_e(a))=p(a)$ on the moved set; outside it
the equality is automatic.  If the test fails, increment $B$.

Deletion reverses these operations using the cached moved set and violation
flag.  If a counter changes from $1$ to $0$, add $p(a)$ to $W$.  Hence the
invariants defining $F$, $W$, and $B$ hold after every update.  When $B=0$,
$p$ is invariant under every active generator and therefore under the current
holonomy subgroup; \cref{thm:main} gives $\CF(e)=1-W$.
\end{proof}

\begin{corollary}[Online increment identity]
When a compatible chord is inserted, the decrement of $W$ performed by the
data structure is exactly
\[
p\bigl(F_{\mathrm{old}}\setminus\Fix(h_e)\bigr),
\]
so the update realizes \eqref{eq:increment} without recomputing the old
holonomy intersection.
\end{corollary}

\subsection{Why tree-edge updates are different}

Changing a tree-edge gain changes $\tau_v$ for every descendant of that edge
and therefore changes every chord holonomy having an endpoint in the affected
subtree.  Deleting a tree edge also requires replacement-tree maintenance.
In the worst case, $\Theta(m)$ cached holonomies may change, so the counter
argument above no longer gives an $O(k)$ update.  Dynamic trees and top-tree
methods support path aggregation for group-valued data
\cite{SleatorTarjan1983,HolmLichtenbergThorup2001}, but maintaining an
intersection of fixed sets under the induced global conjugations requires a
separate analysis.

\begin{openproblem}[Fully dynamic gain updates]
Can $p(\Fix(H))$ be maintained under arbitrary edge insertions, deletions, and
gain changes in $o(km)$ worst-case or amortized update time, while preserving
exact compatibility detection?
\end{openproblem}

\section{Compatibility-Preserving CSP Encoding and the Support Dichotomy}
\label{sec:csp}

A general binary relation does not automatically define compatible uniform
edge distributions.  To transfer the CSP dichotomy without hiding this issue,
we isolate the exact compatibility condition needed by the construction.

The general relationship between contextuality, global sections, and
constraint satisfaction is established in earlier work
\cite{Abramsky2013,AbramskyGottlobKolaitis2013,Simmons2018}. The issue
addressed here is more specific: an arbitrary relational CSP support need not
admit probability tables with consistent one-variable marginals. The
contribution of this section is a compatibility-preserving probabilistic
encoding for a precisely defined class of binary languages and an exact
transfer of the finite-domain dichotomy to the resulting promised
empirical-model problem.

\begin{definition}[Common-marginal realizable language]
A finite binary constraint language $\Gamma$ on $O$ is
\emph{common-marginal realizable} if there exist a distribution
$\pi\in\D(O)$ and, for each $R\in\Gamma$, a distribution
$\mu_R\in\D(O\times O)$ such that
\[
\supp(\mu_R)=R,
\qquad
\mu_R|_1=\mu_R|_2=\pi.
\]
The realization $(\pi,\{\mu_R\}_{R\in\Gamma})$ is fixed as part of the
language.
\end{definition}

Regular relations under the uniform distribution are examples: permutation relations, equality, disequality, and every binary relation whose bipartite support is regular on both sides admit such realizations.  The condition is not automatic for arbitrary relations.  It constrains the probability tables used to form a compatible empirical model, while leaving the logical CSP support unchanged.  Thus the transfer below is exhaustive within the common-marginal realizable class, not within all probabilistic encodings of all binary languages.

For a fixed common-marginal realization, let $\textsc{CF-Support}(\Gamma)$ denote the promised decision problem whose inputs are edge-labelled empirical models assembled from the fixed tables $\{\mu_R:R\in\Gamma\}$ and the equality table $\mu_{=}$ by the construction below, and whose question is whether $\CF<1$.

An arbitrary binary CSP instance may contain several constraint occurrences sharing the same variable pair. To obtain an ordinary edge cover without losing compatibility, split variable occurrences.  For every incidence of a variable $x$ in a constraint $C$, create a copy $x_C$.  For a constraint $C=R(x,y)$, place the table $\mu_R$ on the edge $(x_C,y_C)$.  For each original variable $x$, connect all copies $\{x_C:C\ni x\}$ by a tree of equality edges, using
\[
\mu_{=}(a,b):=\pi(a)\1[a=b].
\]
Every local table now has marginal $\pi$ at both endpoints, and each edge has a unique scope.

\begin{theorem}[CSP--contextual-support equivalence]
\label{thm:csp-transfer}
Let $\Gamma$ be common-marginal realizable and let $e^I$ be the empirical model obtained from a binary $\CSP(\Gamma)$ instance $I$ by occurrence splitting and equality trees.  Then $e^I$ is compatible and
\begin{equation}
\label{eq:csp-equivalence}
I\text{ is satisfiable}
\quad\Longleftrightarrow\quad
\NCF(e^I)>0
\quad\Longleftrightarrow\quad
\CF(e^I)<1.
\end{equation}
Moreover,
\[
\CSP(\Gamma)
\equiv_{\mathrm p}
\textsc{CF-Support}(\Gamma)
\equiv_{\mathrm p}
\CSP(\Gamma\cup\{=\}),
\]
where the equivalences are polynomial-time many-one reductions on the stated promised input representation.
\end{theorem}

\begin{proof}
All relation and equality tables have marginal $\pi$ at both endpoints, so the model is compatible.

If $g$ satisfies the original CSP instance, assign every occurrence copy $x_C$ the value $g(x)$.  This assignment satisfies every relation edge and every equality edge.  Every selected relation event has positive probability. Moreover, because $\supp(\mu_R)=R$, every value appearing in a selected relation tuple has positive $\pi$-mass, so every selected equality event also has positive probability. Letting $\rho$ be the minimum of these finitely many positive probabilities; letting $\rho$ be the minimum of these finitely many positive probabilities and assigning primal weight $\rho$ to this global column gives $\NCF(e^I)>0$.

Conversely, if $\NCF(e^I)>0$, some global column has positive primal weight.
It cannot select a zero-probability event.  Equality edges therefore force all
copies of each original variable to have the same value, and relation edges
then show that those values satisfy every original constraint.  The number of
copy variables and equality edges is linear in the number of constraint
incidences.  On models produced by the construction, contracting each equality
tree recovers the original instance.
\end{proof}

\begin{corollary}[Transferred finite-domain dichotomy]
\label{cor:csp-dichotomy}
For every fixed common-marginal realizable finite language $\Gamma$, deciding
whether $\CF(e^I)<1$ is either polynomial-time solvable or NP-complete.  The
side of the dichotomy is exactly the side occupied by $\CSP(\Gamma)$ in the
Bulatov--Zhuk classification \cite{Bulatov2017,Zhuk2020}.  Algebraically,
after the standard passage to a core and addition of constants, the tractable
side is characterized by a Taylor, equivalently weak near-unanimity,
polymorphism.
\end{corollary}

\begin{example}[The tractable and hard boundaries]
For a language consisting of permutation relations, the uniform distribution
is a common marginal, and satisfiability reduces to propagation along a
spanning forest followed by cycle-consistency checks. This is the tractable
gain-graph side developed in \cref{sec:model,sec:static}.

In contrast, for $O=\{1,2,3\}$ and $\Gamma=\{\neq\}$, the uniform
distribution on ordered unequal pairs is a common-marginal realization and
$\CSP(\Gamma)$ is graph $3$-colourability. Hence the corresponding
$\textsc{CF-Support}(\Gamma)$ problem is NP-complete, recovering the hard
example in \cref{prop:general-complexity}.
\end{example}

\begin{remark}[Boundary of the transfer]
The dichotomy concerns the support threshold $\CF<1$ versus $\CF=1$.  It does
not classify exact computation or approximation of $\CF$ in the tractable CSP
cases.  It also does not apply to arbitrary probability tables lacking a
common-marginal realization.
\end{remark}

\section{Discussion, Scope, and Open Problems}
\label{sec:scope}

The contribution is computational rather than a reinvention of gain-graph
holonomy.  The chain
\[
\text{contextual-fraction LP}
\longrightarrow
\text{holonomy fixed-set mass}
\longrightarrow
\text{query-optimal exact algorithm}
\]
is specific to edge supports that are graphs of bijections generated from one root distribution.  The lower bound applies to deterministic exact algorithms on explicit permutation tables and to cyclic graphs; randomized approximation and compressed group representations may have different complexity.

From the viewpoint of structured computation, the main mechanism is an exact reduction of an exponentially indexed linear program to a graph traversal and an intersection of finite fixed sets. This is a structure-exploiting optimization result: symmetry and transport constraints replace global enumeration by a low-dimensional algebraic invariant. The matching query lower bound further shows that, under explicit table storage, the remaining linear scan is information-theoretically unavoidable.

The chord-dynamic result is similarly precise. It assumes a fixed spanning tree and updates only non-tree edges.  A tree-edge update can change the transport of an entire subtree and therefore many cached holonomies.  Whether this interaction can be handled with link--cut or top-tree methods without recomputing a linear number of fixed-set constraints remains open.

The CSP transfer classifies only the support threshold $\CF<1$ versus $\CF=1$.  It does not classify exact values or additive approximation in the tractable CSP cases.  The common-marginal condition is a genuine probabilistic restriction: it contains permutation, equality, disequality, and all regular binary relations under uniform weighting, but not every binary relation.

\section{Conclusion}
\label{sec:conclusion}

For probability-weighted permutation gain graphs,
\[
\NCF(e)=p(\Fix(H)),\qquad \CF(e)=1-p(\Fix(H)).
\]
This identity removes the $|O|^{|V|}$ global-assignment representation. The resulting exact algorithm is worst-case query optimal on $2$-edge-connected explicit-table inputs, and its value can be maintained in $O(|O|)$ time under fixed-tree chord updates.  Under relational relaxation, the support threshold inherits the finite-domain CSP dichotomy within the common-marginal realizable class.

The paper therefore identifies a rigorous computational tractability island: a generally NP-complete support problem with an exponentially indexed exact linear program becomes quantitatively solvable in linear arithmetic time, optimal with respect to explicit table access, and maintainable under a natural chord-update model.

\appendix
\section{A Concrete State-Space Reduction}
\label{app:example}

Let $G_n=C_n$, let $O=\{1,\ldots,8\}$, put identity gains on a spanning path,
and put the cycle $(1\ 2\ 3)$ on the closing edge.  For uniform $p$,
\[
F=\{4,5,6,7,8\},\qquad
\NCF(e_n)=\frac58,\qquad
\CF(e_n)=\frac38.
\]
The general LP has $8^n$ global columns, while the structural calculation
scans one length-$8$ nontrivial holonomy after a linear graph traversal.

\begin{center}
\begin{tabular}{@{}r r r@{}}
\toprule
$n$ & Global assignments $8^n$ & Scale $8(|V|+|E|)=16n$\\
\midrule
$10$ & $1{,}073{,}741{,}824$ & $160$\\
$25$ & $3.78\times10^{22}$ & $400$\\
$50$ & $1.43\times10^{45}$ & $800$\\
$100$ & $2.04\times10^{90}$ & $1600$\\
\bottomrule
\end{tabular}
\end{center}

The last column records the scale in the proved asymptotic bound; it is not an exact instruction count.  The example demonstrates combinatorial state-space collapse, while \cref{thm:static-lower} shows that the remaining linear scan of explicit edge tables is unavoidable in the worst case.

%
%

\bibliographystyle{unsrtnat}
\bibliography{refs}

@article{Abramsky2013,
	author  = {Abramsky, Samson},
	title   = {Relational Hidden Variables and Non-Locality},
	journal = {Studia Logica},
	volume  = {101},
	number  = {2},
	pages   = {411--452},
	year    = {2013},
	doi     = {10.1007/s11225-013-9477-4},
	url     = {https://doi.org/10.1007/s11225-013-9477-4}
}

@inproceedings{AbramskyGottlobKolaitis2013,
	author    = {Abramsky, Samson and Gottlob, Georg and Kolaitis, Phokion G.},
	title     = {Robust Constraint Satisfaction and Local Hidden Variables in Quantum Mechanics},
	booktitle = {Proceedings of the Twenty-Third International Joint Conference on Artificial Intelligence},
	pages     = {440--446},
	year      = {2013}
}

@article{Simmons2018,
	author  = {Simmons, Adam W.},
	title   = {On the Computational Complexity of Detecting Possibilistic Locality},
	journal = {Journal of Logic and Computation},
	volume  = {28},
	number  = {1},
	pages   = {203--217},
	year    = {2018},
	doi     = {10.1093/logcom/exx045},
	url     = {https://doi.org/10.1093/logcom/exx045}
}

@article{AbramskyBrandenburger2011,
	author  = {Abramsky, Samson and Brandenburger, Adam},
	title   = {The Sheaf-Theoretic Structure of Non-Locality and Contextuality},
	journal = {New Journal of Physics},
	volume  = {13},
	number  = {11},
	pages   = {113036},
	year    = {2011},
	doi     = {10.1088/1367-2630/13/11/113036},
	url     = {https://doi.org/10.1088/1367-2630/13/11/113036}
}

@article{AbramskyBarbosaMansfield2017,
	author  = {Abramsky, Samson and Barbosa, Rui Soares and Mansfield, Shane},
	title   = {The Contextual Fraction as a Measure of Contextuality},
	journal = {Physical Review Letters},
	volume  = {119},
	number  = {5},
	pages   = {050504},
	year    = {2017},
	doi     = {10.1103/PhysRevLett.119.050504},
	url     = {https://doi.org/10.1103/PhysRevLett.119.050504}
}

@inproceedings{Bulatov2017,
	author    = {Bulatov, Andrei A.},
	title     = {A Dichotomy Theorem for Nonuniform {CSP}s},
	booktitle = {Proceedings of the 58th IEEE Symposium on Foundations of Computer Science},
	series    = {FOCS 2017},
	pages     = {319--330},
	publisher = {IEEE},
	year      = {2017},
	doi       = {10.1109/FOCS.2017.37},
	url       = {https://doi.org/10.1109/FOCS.2017.37}
}

@book{GrossTucker1987,
	author    = {Gross, Jonathan L. and Tucker, Thomas W.},
	title     = {Topological Graph Theory},
	publisher = {John Wiley \& Sons},
	address   = {New York},
	year      = {1987}
}

@article{HolmLichtenbergThorup2001,
	author  = {Holm, Jacob and de Lichtenberg, Kristian and Thorup, Mikkel},
	title   = {Poly-Logarithmic Deterministic Fully-Dynamic Algorithms for Connectivity, Minimum Spanning Tree, 2-Edge, and Biconnectivity},
	journal = {Journal of the ACM},
	volume  = {48},
	number  = {4},
	pages   = {723--760},
	year    = {2001},
	doi     = {10.1145/502090.502095},
	url     = {https://doi.org/10.1145/502090.502095}
}

@misc{Montanhano2021,
	author        = {Montanhano, Samuel Barbosa},
	title         = {Contextuality in the Bundle Approach, {$n$}-Contextuality, and the Role of Holonomy},
	year          = {2021},
	eprint        = {2105.14132},
	archiveprefix = {arXiv},
	primaryclass  = {quant-ph},
	note          = {Revised 2024},
	doi           = {10.48550/arXiv.2105.14132},
	url           = {https://doi.org/10.48550/arXiv.2105.14132}
}

@article{SleatorTarjan1983,
	author  = {Sleator, Daniel D. and Tarjan, Robert E.},
	title   = {A Data Structure for Dynamic Trees},
	journal = {Journal of Computer and System Sciences},
	volume  = {26},
	number  = {3},
	pages   = {362--391},
	year    = {1983},
	doi     = {10.1016/0022-0000(83)90006-5},
	url     = {https://doi.org/10.1016/0022-0000(83)90006-5}
}

@article{Vorobev1962,
	author  = {Vorob'ev, Nikolai Nikolaevich},
	title   = {Consistent Families of Measures and Their Extensions},
	journal = {Theory of Probability and Its Applications},
	volume  = {7},
	number  = {2},
	pages   = {147--163},
	year    = {1962},
	doi     = {10.1137/1107014},
	url     = {https://doi.org/10.1137/1107014}
}

@article{Zaslavsky1989,
	author  = {Zaslavsky, Thomas},
	title   = {Biased Graphs. {I}. Bias, Balance, and Gains},
	journal = {Journal of Combinatorial Theory, Series B},
	volume  = {47},
	number  = {1},
	pages   = {32--52},
	year    = {1989},
	doi     = {10.1016/0095-8956(89)90063-4},
	url     = {https://doi.org/10.1016/0095-8956(89)90063-4}
}

@article{Zaslavsky2009,
	author  = {Zaslavsky, Thomas},
	title   = {Totally Frustrated States in the Chromatic Theory of Gain Graphs},
	journal = {European Journal of Combinatorics},
	volume  = {30},
	number  = {1},
	pages   = {133--156},
	year    = {2009},
	doi     = {10.1016/j.ejc.2008.02.004},
	url     = {https://doi.org/10.1016/j.ejc.2008.02.004}
}

@article{ZaslavskyGlossary,
	author  = {Zaslavsky, Thomas},
	title   = {A Mathematical Bibliography of Signed and Gain Graphs and Allied Areas},
	journal = {Electronic Journal of Combinatorics},
	number  = {DS8},
	year    = {1998},
	note    = {Dynamic Survey; revised editions},
	doi     = {10.37236/27},
	url     = {https://doi.org/10.37236/27}
}

@article{Zhuk2020,
	author  = {Zhuk, Dmitriy},
	title   = {A Proof of the {CSP} Dichotomy Conjecture},
	journal = {Journal of the ACM},
	volume  = {67},
	number  = {5},
	articleno = {30},
	pages   = {1--78},
	year    = {2020},
	doi     = {10.1145/3402029},
	url     = {https://doi.org/10.1145/3402029}
}
\end{document}